%% file: main.tex
\documentclass[final]{IEEEtran}

\usepackage{cite}
\usepackage{amsthm, amsmath,amssymb,amsfonts,bbm}
\usepackage{algorithmic}
\usepackage[dvips]{graphicx}
\usepackage{subfigure}
\graphicspath{{graphics/}, {graphics/Results}, {graphics/Results/0_MatlabResults}}
\usepackage{makecell}
\usepackage{multirow}
\usepackage{filecontents, pgfplots}
\usepackage[algoruled]{algorithm2e}
\usepackage{microtype}
\usepackage{tikz}
\usepackage{textcomp}
\usepackage{xcolor}
\usepackage{verbatim}
\usepackage[binary-units=true]{siunitx}
\DeclareSIUnit{\belmilliwatt}{Bm}
\DeclareSIUnit{\dBm}{\deci\belmilliwatt}
\def\BibTeX{{\rm B\kern-.05em{\sc i\kern-.025em b}\kern-.08em
		T\kern-.1667em\lower.7ex\hbox{E}\kern-.125emX}}

\makeatletter
\newif\iftag@here

\newcommand*{\taghere}[1][0pt]
{\ifmeasuring@\else
	\global\tag@heretrue
	\tikz[remember picture,overlay]{\coordinate (taghere) at (0pt,#1);}%
	\fi}

\def\place@tag{%
	\iftagsleft@
	\kern-\tagshift@
	\iftag@here
	\global\tag@herefalse
	\tikz[remember picture,overlay]%
	{\path (taghere) -| node[anchor=base]{\rlap{\boxz@}} (0pt,0pt);}%
	\else
	\if1\shift@tag\row@\relax
	\rlap{\vbox{%
			\normalbaselines
			\boxz@
			\vbox to\lineht@{}%
			\raise@tag
	}}%
	\else
	\rlap{\boxz@}%
	\fi
	\kern\displaywidth@
	\fi
	\else
	\kern-\tagshift@
	\iftag@here
	\global\tag@herefalse
	\tikz[remember picture,overlay]%
	{\path  (taghere) -|  node[anchor=base]{\llap{\boxz@}} (0pt,0pt);}%
	\else
	\if1\shift@tag\row@\relax
	\llap{\vtop{%
			\raise@tag
			\normalbaselines
			\setbox\@ne\null
			\dp\@ne\lineht@
			\box\@ne
			\boxz@
	}}%
	\else \llap{\boxz@}%
	\fi
	\fi
	\fi
}
\makeatother

\usepackage[acronym,nomain]{glossaries}
\input{Glossary}

\begin{document}

	\newtheorem{proposition}{Proposition}	
	\newtheorem{lemma}{Lemma}	
	\newtheorem{corollary}{Corollary}
	\newtheorem{assumption}{Assumption}	
	\newtheorem{remark}{Remark}	
	
	\title{	EH Modelling and Achievable Rate for FSO SWIPT Systems with Non-linear Photovoltaic Receivers}

	\author{\IEEEauthorblockN{Nikita Shanin, Hedieh Ajam, 
			Vasilis K. Papanikolaou, Bernhard Schmauss, 
			Laura Cottatellucci, and Robert Schober}\\
		\IEEEauthorblockA{\textit{Friedrich-Alexander-Universit\"{a}t Erlangen-N\"{u}rnberg (FAU), Germany} }		\vspace*{-13pt}}

	\maketitle
	\thispagestyle{empty}

\begin{abstract}
	\input{Abstract}
\end{abstract}

\vspace*{-13pt}
\section{Introduction}
\label{Section:Introduction}
\input{IntroductionSection}

\vspace*{-10pt}
\section{System Model}
\label{Section:SysModel}
\input{SystemModelSection}

\vspace*{-10pt}
\section{EH at the Photovoltaic RX}
\label{Section:EH}
\input{EH_ID}

\vspace*{-10pt}
\section{Information Reception at the RX}
\label{Section:ID}
\input{InfDecodingSection}

\vspace*{-12pt}
\section{Numerical Results}
	\label{Section:NumResults}
\input{SimulationResults}

\vspace*{-12pt}
\section{Conclusions}
	\label{Section:Conclusions}
	\input{ConclusionsSection}
\vspace*{-12pt}
\appendices
	\renewcommand{\thesection}{\Alph{section}}
	\renewcommand{\thesubsection}{\thesection.\arabic{subsection}}
	\renewcommand{\thesectiondis}[2]{\Alph{section}:}
	\renewcommand{\thesubsectiondis}{\thesection.\arabic{subsection}:}	
	\section{Proof of Proposition \ref{Theorem:Capacity}}
	\label{Appendix:ProofCapacity}
	\input{ProofCapacity}

\vspace*{-17pt}
\bibliographystyle{IEEEtran}
\bibliography{WPT_Bibl.bib}

\end{document}

%% file: Glossary.tex
\makeglossaries
\newacronym{swipt}{SWIPT}{simultaneous wireless information and power transfer}
\newacronym{wpt}{WPT}{wireless power transfer}
\newacronym{wpcn}{WPCN}{wireless powered communication network}
\glsunset{wpcn}
\newacronym{wit}{WIT}{wireless information transfer}
\newacronym{awgn}{AWGN}{additive white Gaussian noise}
\newacronym{tx}{TX}{transmitter}
\glsunset{tx}
\newacronym{rx}{RX}{receiver}
\glsunset{rx}
\newacronym{bs}{BS}{base station}
\newacronym{ir}{IR}{information receiver}
\newacronym{eh}{EH}{energy harvesting}
\newacronym{id}{ID}{information detection}
\newacronym{irs}{IRS}{intelligent reflected surface}
\newacronym{ap}{AP}{average power}
\newacronym{pp}{PP}{peak power}
\newacronym{aoa}{AOA}{angle-of-arrival}
\newacronym{aod}{AOD}{angle-of-departure}
\newacronym{eec}{EEC}{electical equivalent circuit}

\newacronym{fso}{FSO}{free space optical}
\glsunset{fso}
\newacronym{vlc}{VLC}{visible light communication}
\newacronym{led}{LED}{light emitting diode}
\glsunset{led}
\newacronym{ook}{OOK}{On-Off keying}
\newacronym{pam}{PAM}{pulse amplitude modulated}
\newacronym{tdma}{TDMA}{time division multiple access}
\newacronym{ml}{ML}{maximum likelihood}

\newacronym{siso}{SISO}{single-input single-output}
\newacronym{mimo}{MIMO}{multiple-input multiple-output}
\newacronym{miso}{MISO}{multiple-input single-output}
\newacronym{simo}{SIMO}{single-input multiple-output}

\newacronym{rf}{RF}{radio frequency}
\newacronym{dc}{DC}{direct current}
\newacronym{ac}{AC}{alternating current}
\newacronym{papr}{PAPR}{peak-to-average power ratio}
\newacronym{lp}{LPF}{low-pass filter}
\newacronym{mc}{MC}{matching circuit}
\newacronym{mrt}{MRT}{maximum ratio transmission}
\newacronym{ecb}{ECB}{equivalent complex baseband}
\newacronym{tdd}{TDD}{time-division-duplex}
\newacronym{zf}{ZF}{zero forcing}
\newacronym{snr}{SNR}{signal-to-noise ratio}
\newacronym{sinr}{SINR}{signal-to-interference-plus-noise ratio}

\newacronym{rv}{RV}{random variable}
\newacronym{iid}{i.i.d.}{independent and identically distributed}
\newacronym{Pdf}{Pdf}{Probability density function}
\newacronym{pdf}{pdf}{probability density function}
\newacronym{cdf}{cdf}{cumulative density function}

\newacronym{dnn}{DNN}{dense neural networks}
\glsunset{dnn}
\newacronym{mdp}{MDP}{Markov decision process}

\newacronym{sca}{SCA}{successive convex approximation}
\newacronym{sdr}{SDR}{semi-definite relaxation}

\newacronym{spr}{LP}{low power}
\newacronym{mpr}{MP}{medium power}
\newacronym{lpr}{HP}{high power}

%% file: Abstract.tex
In this paper, we study optical simultaneous wireless information and power transfer (SWIPT) systems, where a photovoltaic optical receiver (RX) is illuminated by ambient light and an intensity-modulated free space optical (FSO) signal.
To facilitate simultaneous information reception and energy harvesting (EH) at the RX, the received optical signal is first converted to an electrical signal, and then, its alternating current (AC) and direct current (DC) components are separated and utilized for information decoding and EH, respectively.
By accurately analysing the equivalent electrical circuit of the photovoltaic RX, we model the current flow through the photovoltaic p-n junction in both the low and high input power regimes using a two-diode model of the p-n junction and we derive a closed-form non-linear EH model that characterizes the harvested power at the RX.
Furthermore, taking into account the non-linear behaviour of the photovoltaic RX on information reception, we derive the optimal distribution of the transmit information signal that maximizes the achievable information rate.
The proposed EH model is validated by circuit simulation results.
Furthermore, we compare with two baseline models based on maximum power point (MPP) tracking at the RX and a single-diode p-n junction model, respectively, and demonstrate that in contrast to the proposed EH model, they are not able to fully capture the non-linearity of photovoltaic optical RXs.
Finally, our numerical results highlight that the proposed optimal distribution of the transmit signal yields significantly higher achievable information rates compared to uniformly distributed transmit signals, which are optimal for linear optical information RXs.

%% file: IntroductionSection.tex
Free space optical (FSO) communication systems have attracted significant attention in the past decades thanks to their ability to provide extremely high data rates by focusing an intensity-modulated laser signal on a small optical receiver (RX) \cite{Zhang2019, Ajam2020}.
Furthermore, since practical RXs in FSO systems employ light-sensitive devices, e.g., photodetectors and photovoltaic cells, to convert the received optical signal into an electrical signal, they are also able to harvest energy from the received signal \cite{Mertens2014, Luque2010}.
This aspect has recently fuelled interest in FSO \gls*{swipt} systems, where the laser signal is exploited not only to convey information, but also to deliver power to user devices \cite{MaShuai2019, Tran2019, Sepehrvand2021, Wang2015, Fakidis2020}.

Optical SWIPT systems employing RXs equipped with \gls*{eh} and information decoding circuits were considered in \cite{MaShuai2019, Tran2019, Sepehrvand2021}, where light emitting diodes (LEDs) were employed at the transmitter (TX) to cover a large area and simultaneously provide indoor illumination and deliver data and energy to user devices.
Since photovoltaic cells and photodetectors are preferable for EH and information reception, respectively, the authors in \cite{MaShuai2019} proposed to deploy both types of optical RXs at the user device for SWIPT.
However, the surface area of practical user equipments may be constrained to be small, and thus, it may not be possible to mount multiple optical RXs on the same user device.
Therefore, in \cite{Tran2019}, the authors studied hybrid RF-optical SWIPT systems and proposed an optical RX design based on a single photovoltaic cell.
In particular, the photovoltaic device was utilized to convert the received optical signal into an electrical signal, whose \gls*{dc} and \gls*{ac} components were separated and utilized for EH and information decoding at a user device, respectively.
For the optical SWIPT systems in \cite{MaShuai2019} and \cite{Tran2019}, EH at the user device was designed by tracking the maximum power point (MPP) at the RX.
In other words, the EH load parameters in \cite{MaShuai2019} and \cite{Tran2019} were tuned to maximize the EH efficiency for a given received power, which, in general, requires EH loads with variable impedances, and thus, is not feasible with practical RX designs, where the electrical properties of the EH load device are typically fixed \cite{Tietze2012}.
To avoid this impractical assumption, in \cite{Sepehrvand2021}, the authors derived a non-linear EH model characterizing the instantaneous power harvested at the RX as a function of the instantaneous received power.

Since the transmit power of optical light sources is limited by hardware constraints, the power received at small photovoltaic cells is typically rather low if undirected LEDs are employed at the TX.
Therefore, in \cite{Wang2015} and \cite{Fakidis2020}, the authors proposed to exploit a directed FSO laser beam focused on the RX to transfer more power to user devices.
In particular, separating the AC and DC signal components at the RX and assuming a fixed EH load resistance, the authors of \cite{Wang2015} derived not only a non-linear EH model for the photovoltaic cell at the RX, but also the frequency response of the RX for the received information signal and the noise impairing information reception.
Furthermore, experimental results in \cite{Fakidis2020} confirmed that for SWIPT with photovoltaic RXs, the non-linearities of photovoltaic cells have to be carefully taken into account for efficient system design, and demonstrated that extremely high data rates exceeding $\SI{1}{\giga\bit\per\second}$ are possible.

We note that although the FSO SWIPT designs in \cite{Wang2015, Fakidis2020} were shown to achieve high data rates, they were still based on several assumptions that may limit performance.
First, in \cite{MaShuai2019, Tran2019, Sepehrvand2021, Wang2015, Fakidis2020}, the authors modelled the photovoltaic cell by an \gls*{eec} employing a single non-linear diode.
However, for high received signal powers, the output electrical current of photovoltaic cells is mainly determined by the diffusion of electrons and holes in the p-n junction of the cells, while for low input powers, diffusion is negligible and the photovoltaic current is caused by particle recombination in the depletion region of the junction \cite{Luque2010, Mertens2014}.
Thus, an accurate \gls*{eec} of a photovoltaic cell has to comprise two diodes to model the non-linearities of the output current in both the low and high input power regimes \cite{Mertens2014, Abbassi2018}.
Furthermore, for the RX design in \cite{MaShuai2019, Tran2019, Sepehrvand2021, Wang2015, Fakidis2020}, the authors assumed a linear frequency response of the RX circuit with respect to the received information signal, which may not be realistic due to the non-linear behaviour of photovoltaic cells.
Moreover, for the analysis of SWIPT systems in \cite{Wang2015, Fakidis2020}, the authors neglected the impact of ambient light on the EH and information reception at the photovoltaic RX.
Finally, we note that the design of the optimal filter for the received information signal is not feasible for practical photovoltaic RX designs due to the RX non-linearities and the extremely high sampling rates required for accurate discretization of the output information signal \cite{Fakidis2020, Tietze2012}.
To the best of the authors' knowledge, the accurate modelling and design of photovoltaic FSO SWIPT systems, where all non-linearities of the photovoltaic RXs are taken into account, have not been considered in the literature, yet.
\begin{figure}[!t]
	\centering
	\includegraphics[draft=false, width=0.85\linewidth]{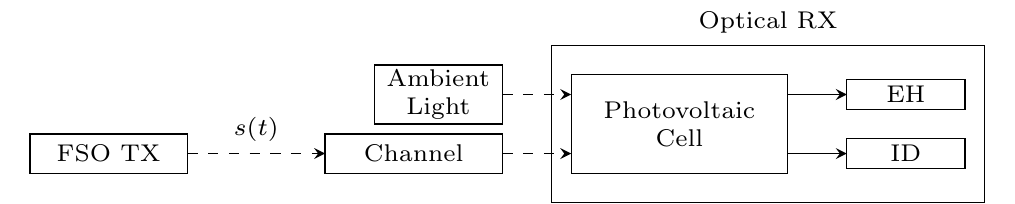}
	\vspace*{-5pt}
	\caption{FSO SWIPT system comprising an FSO TX and a photovoltaic optical RX.}
	\vspace*{-15pt}
	\label{Fig:SystemModel}
\end{figure}

In this work, we consider an optical \gls*{swipt} system with a photovoltaic optical RX.
The photovoltaic RX is illuminated by an information-carrying FSO signal and ambient light and converts the received optical signal into an output electrical signal, whose DC and AC components are separated and utilized for EH and information decoding, respectively.
The main contributions of this work can be summarized as follows.
	First, taking into account the non-linearities of the RX in both the high and low input power regimes, we derive an accurate closed-form EH model to characterize the instantaneous harvested power at the RX.
	Next, since optimal filtering of the electrical information signal is not possible due to the non-linearity of photovoltaic RXs and the high required sampling rates, we design a suboptimal information decoder.
	Furthermore, taking into account the non-linear behaviour of the photovoltaic RX, we derive the optimal transmit signal distribution maximizing the achievable information rate for FSO SWIPT systems.
	We validate the proposed EH model by circuit simulation results.
	Furthermore, we demonstrate that in contrast to the proposed EH model, two baseline EH models based on MPP tracking at the RX and a single-diode p-n junction model, respectively, are not able to accurately capture the non-linearity of photovoltaic RX circuits.
	Finally, our numerical results highlight that the proposed optimal distribution of the TX signal yields significantly higher achievable information rates compared to uniformly distributed TX signals, which are optimal for linear optical information RXs \cite{Lapidoth2009}.


{\itshape Notations:}
The sets of real, real non-negative, and non-negative natural numbers are represented by $\mathbb{R}$, $\mathbb{R}_{+}$, and $\mathbb{N}$, respectively.
The domain of one-dimensional function $f(\cdot)$ is denoted by $\text{dom} \{f\}$.
$\delta(\cdot)$ is the Dirac delta function.
Furthermore, $f_s(\cdot)$ and $F_s(\cdot)$ denote the \gls*{pdf} and \gls*{cdf} of random variable $s$, respectively.
Finally, $\mathbb{E}\{s\}$ stands for the statistical expectation of random variable $s$.

%% file: SystemModelSection.tex
We consider an optical \gls*{swipt} system, where an FSO \gls*{tx} sends an intensity-modulated FSO signal focused on an optical \gls*{rx} equipped with a photovoltaic cell \cite{Mertens2014, Luque2010, Wang2015, Fakidis2020}, as shown in Fig.~\ref{Fig:SystemModel}.
We denote the carrier wavelength of the transmit signal by $\lambda_0$ and express the power spectral density of the stochastic process modelling the received light as a function of wavelength $\lambda$ and time $t, t \in [kT, (k+1)T)$, as follows:
\begin{equation}
	p_\text{r}(\lambda, t) = h s(t)  \delta(\lambda - \lambda_0)  + \tilde{p}_\text{a}(\lambda) + \tilde{w}(\lambda, t).
	\label{Eqn:ReceivedSpectrum}
\end{equation} 
Here, $s(t) = s[k] \psi(t-kT)$ is the power of the intensity-modulated optical signal, where $s[k], \forall k \in \mathbb{N},$ are information symbols transmitted in time slots $k, k\in \mathbb{N},$ and modelled as \gls*{iid} realizations of the non-negative random variable $s$, and $\psi(t)$ is a rectangular pulse of duration $T$, i.e., $\psi(t)$ takes value 1 if $t\in[0, T)$ and $0$, otherwise.
Since the intensity of the FSO TX signal is limited by hardware and eye safety constraints \cite{Ajam2020, Lapidoth2009, Lasersafety2014}, we limit the TX power by $A^2$, and thus, the pdf of $s$ satisfies $\text{dom} \{f_s\} \subseteq [0, A^2]$.
In (\ref{Eqn:ReceivedSpectrum}), $h$ and $\tilde{p}_\text{a}(\lambda)$ are the channel gain between TX and RX and the power spectral density of the ambient light, respectively, and are assumed to be known at both devices.
Finally, $\tilde{w}(\lambda, t)$ is the time-varying power spectral density of the noise received at the photovoltaic RX.
We note that the input RX noise is a non-stationary stochastic process caused by the random fluctuations of the ambient light and the intensity of the low-cost laser source \cite{Lapidoth2009}.
Furthermore, the output signal at the photovoltaic RX is additionally impaired by thermal and shot noise, which are generated by the resistances and p-n junctions of the RX, respectively \cite{Lapidoth2009, Wang2015}.
For the derivation of the EH model in Section~\ref{Section:EH}, we neglect the RX noise since its contribution to the average harvested power is negligible \cite{MaShuai2019, Tran2019, Sepehrvand2021, Wang2015}.
To analyse the performance of information transmission, in Section~\ref{Section:ID}, we assume that the thermal noise at the RX dominates and we model the equivalent noise, which impairs the output information symbols, as \gls*{awgn} \cite{Wang2015, Lapidoth2009}.

%% file: EH_ID.tex
In this section, we study the \gls*{eh} at the photovoltaic RX.
To this end, we first present the \gls*{eec} of the RX, where the AC and DC components of the electrical output signal are separated and utilized for information reception and EH, respectively.
Next, we propose a model for characterizing the EH at the user device.

\paragraph*{EEC of the RX} In the following, we present the EEC of the photovoltaic cell employed at the RX and shown in Fig.~\ref{Fig:MjCell}.
First, we express the current $j(t)$ induced in the photovoltaic cell by the received light as follows \cite{Mertens2014}:\vspace*{-3pt}
\begin{equation}
	j(t) = \int p_\text{r}(\lambda, t) r(\lambda) d\lambda = h r_0 s(t) + {p}_\text{a} + w(t) ,
	\label{Eqn:InducedCurrent}\vspace*{-3pt}
\end{equation}
\noindent where $r(\lambda)$ is the wavelength dependent spectral response of the photovoltaic cell with $r_0 = r(\lambda_0)$ \cite{Mertens2014}, and $p_\text{a} = {\int \tilde{p}_\text{a}(\lambda) r(\lambda)  d\lambda }$ and $w(t) = {\int \tilde{w}(\lambda, t) r(\lambda)  d\lambda }$ are the photovoltaic currents due to the ambient light and the RX noise, respectively.

Next, we note that as shown in Fig.~\ref{Fig:MjCell}, the accurate EEC of a photovoltaic cell comprises two diodes that represent the diffusion current in the neutral region and particle recombination in the depletion region of the photovoltaic p-n junction, respectively \cite{Mertens2014, Luque2010}.
In particular, at low received light intensities, the diffusion current is small compared to the particle recombination in the depletion region and the corresponding recombination current $i^\text{d}_2$ can be modelled by a diode with ideality factor $2$ \cite{Luque2010}.
However, if the received light intensity is high, the diffusion current $i^\text{d}_1$ dominates and can be modelled by a diode with ideality factor $1$ \cite{Luque2010}.
Furthermore, we model the parasitic series and shunt resistances of the p-n junction by resistors $R_\text{s}$ and $R_\text{sh}$ in Fig.~\ref{Fig:MjCell}, respectively.
\begin{figure}[!t]
	\centering
	\includegraphics[draft=false, width=0.75\linewidth]{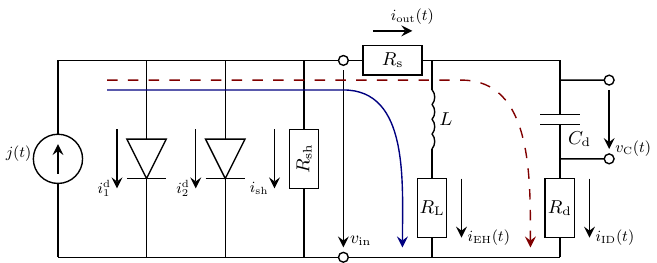}
	\vspace*{-5pt}
	\caption{EEC of the optical photovoltaic RX \cite{Mertens2014}, where the DC (blue solid line) and AC (red dashed line) components of the electrical output signal are separated and utilized for \gls*{eh} and information decoding, respectively \cite{Wang2015, Fakidis2020}.}
	\label{Fig:MjCell}
	\vspace*{-10pt}
\end{figure}

To efficiently harvest power and decode information from the received optical signal, similar to \cite{Tran2019, Wang2015}, we split the flows of the AC and DC signal components via an R-L low-pass filter and an R-C high-pass filter and utilize these signal components for information reception and EH, respectively \cite{Horowitz1989}.
In particular, the DC signal component is harvested at a resistive EH load $R_\text{L}$, whereas the AC component is received at the information RX modelled by resistance $R_\text{d}$, as shown in Fig.~\ref{Fig:MjCell}.
As in \cite{Tran2019, Wang2015}, we assume ideal lossless high- and low-pass electrical filters, i.e., we neglect the parasitic resistances of inductance $L$ and capacitance $C_\text{d}$ and suppose that the DC and AC currents flow only through $R_\text{L}$ and $R_\text{d}$, respectively.

We note that the time interval $T$, capacitance $C_\text{d}$ and inductance $L$, see Fig.~\ref{Fig:MjCell}, can always be tuned to avoid undesired memory effects due to charging and discharging of reactive elements, i.e., capacitances and inductances, in the photovoltaic RX \cite{Horowitz1989}.
Thus, after a transient phase that we assume to be negligible, in each time slot $k, \forall k,$ the photovoltaic RX reaches a steady state and we can neglect the dependence of the output currents $i_\text{out}(t)$, $i_\text{EH}(t)$, and $i_\text{ID}(t)$ in time slot $k, \forall k,$ on information symbols transmitted prior to that time slot, i.e., symbols $s[p]$, $p < k$, $p \in \mathbb{N}$ \cite{Tran2019, Wang2015}, as shown in Fig.~\ref{Fig:InpulseResponse}.

\paragraph*{Derivation of the EH model} In the following, we derive an EH model for the considered photovoltaic cell.
First, as in \cite{MaShuai2019, Tran2019, Sepehrvand2021, Wang2015} and other related works, we neglect the impact of noise on the power harvested at the RX.
Next, since an ideal low-pass R-L filter is assumed for EH, we neglect the ripples of the voltage across the EH load resistance $R_\text{L}$ during the reception of symbol $s[k], \forall k$ \cite{Tietze2012, Sepehrvand2021, Wang2015, Fakidis2020}.
Thus, in steady state, we have $i_\text{ID}(t) \approx 0$ and the equivalent current induced in time slot $k, \forall k$, and the corresponding current flow through the resistance $R_\text{L}$ are given by $j[k] \triangleq j(kT) = r_0 h s[k] + p_\text{a}$ and $i_\text{EH}[k] \triangleq i_\text{EH}(kT) = i_\text{out} (kT)$, respectively.

In the following, for a given power of the FSO received signal\footnotemark\hspace*{0pt} $p = h s \geq 0$, we characterize the corresponding instantaneous harvested power $P_\text{harv}$ at EH load resistance $R_\text{L}$.
\footnotetext{For simplicity of presentation, in this section, we drop the time slot index~$k$.}
To this end, we express the output current $i_\text{EH}$ of the photovoltaic cell as follows \cite{Mertens2014, Tietze2012}:\vspace*{-4pt}
\begin{align}
	i_{\text{EH}} &= j - i^\text{d}_{1} - i^\text{d}_{2}  - i_{\text{sh}} \nonumber\\ 
	&=  j - I_\text{s}^1 (e^{\frac{ v_\text{in} }{V_T}} - 1) - I_\text{s}^2 (e^{\frac{ v_\text{in} }{2 V_T}} - 1) - \frac{ v_\text{in} }{R_{\text{sh}}},\vspace*{-4pt}
	\label{Eqn:OutputCurrent_Junction}
\end{align}
where $I_\text{s}^{n}, n\in\{1,2\}$, and $V_T = \SI{25.85}{\milli\volt}$ are the reverse-bias saturation current of diode $n$ and the thermal voltage of the diode, respectively, and $v_\text{in}$ is the input voltage at the RX, see Fig.~\ref{Fig:MjCell}.
Here, $i^\text{d}_{1}$ and $i^\text{d}_{2}$ model the diffusion current and particle recombination in the depletion region of the p-n junction, respectively \cite{Luque2010}.

\begin{figure}[!t]
	\centering
	\includegraphics[draft=false, width=0.75\linewidth]{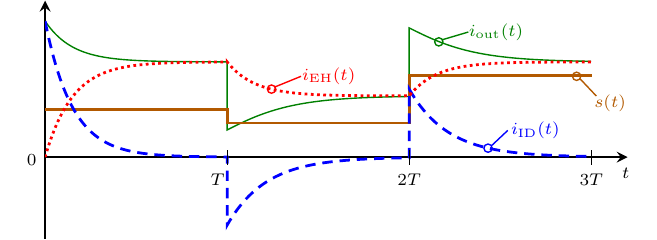} \vspace*{-7pt}
	\caption{Transmit signal $s(t)$ and currents $i_\text{out}(t)$, $i_\text{ID}(t)$, and $i_\text{EH}(t)$ obtained for the EEC in Fig.~\ref{Fig:MjCell} with the circuit parameters specified in Table~\ref{Table:SimulationSetup}.}
	\label{Fig:InpulseResponse}\vspace*{-10pt}
\end{figure}

\begin{table}[!t]
	\centering
	\caption{Parameters of the RX circuit in Fig.~\ref{Fig:MjCell} \cite{Wang2015}}\vspace*{-5pt}
	\begin{tabular}{|l|l|}
		\hline
		$C_\text{d} = \SI{2.5}{\micro\farad}$&
		$R_\text{d} = \SI{10}{\kilo\ohm}$\\
		$L = \SI{10}{\milli\henry}$&
		$R_\text{L} = \SI{10}{\kilo\ohm}$\\
		$R_\text{sh} = \SI{100}{\mega\ohm}$&
		$R_\text{s} = \SI{100}{\ohm}$\\
		\hline
		\multicolumn{2}{|c|}{ \makecell{ $I_\text{s} = I_\text{s}^1 = I_\text{s}^2 = \SI{1}{\nano\ampere}$ } }  
		\\
		\hline
	\end{tabular}
	\label{Table:SimulationSetup}
	\vspace*{-10pt}
\end{table}
Next, we assume that for large shunt resistances $R_\text{sh}$, the current leakage $i_\text{sh}$ is negligible compared to the current flows through the diodes, i.e., $i_\text{sh} \ll i^\text{d}_{n}, \forall n$ \cite{Tietze2012}.
Additionally, we assume that the saturation currents of the diodes modelling the photovoltaic cell are identical\footnotemark, i.e., $I_\text{s}^{1} = I_\text{s}^{2}  = I_\text{s}$.
\footnotetext{If the saturation currents are not identical, one can adopt the current $I_\text{s} \neq I_\text{s}^{1} \neq I_\text{s}^{2}$ that minimizes the expression $ \int_{v_\text{in} \in \mathcal{V}} |I_\text{s}^1 e^{\frac{ v_\text{in} }{V_\text{T}}} + I_\text{s}^2 e^{\frac{ v_\text{in} }{2V_\text{T}}} - I_\text{s} (e^{\frac{ v_\text{in} }{V_\text{T}}} + e^{\frac{ v_\text{in} }{2V_\text{T}}})| \text{d} v_\text{in}$ for the operating range of input voltages $v_\text{in} \in \mathcal{V}$.}
Thus, we can rewrite (\ref{Eqn:OutputCurrent_Junction}) as follows:\vspace*{-3pt}
\begin{equation}
	I_\text{s} e^{\frac{ v_\text{in} }{V_\text{T}}}  +I_\text{s} e^{\frac{ v_\text{in} }{2 V_\text{T}}} + i_{\text{EH}} - j - 2 I_\text{s} = 0.\vspace*{-3pt}
	\label{Eqn:OutCurrentQuadratic}
\end{equation}

Next, we obtain the output voltage $v_\text{in}$ as solution of (\ref{Eqn:OutCurrentQuadratic}), which is quadratic in $e^{\frac{ v_\text{in} }{2 V_\text{T}}}$, as follows:\vspace*{-3pt}
\begin{align}
	v_\text{in} &= \;2 V_\text{T} \ln \bigg( \frac{ - I_\text{s} + \sqrt{ I_\text{s}^2 - 4 I_\text{s} (i_{\text{EH}} - j - 2 I_\text{s} )  } }{ 2 I_\text{s} } \bigg) \nonumber \\
	& \approx V_\text{T} \ln \bigg(  \frac{ j - i_{\text{EH}} }{I_\text{s}} + 1 \bigg),
	\label{Eqn:SolQuadEqn}\vspace*{-3pt}
\end{align}
\noindent where the last approximation is based on $\sqrt{x+2.25} - 0.5 \approx \sqrt{x+1}, \forall x \geq 0$.

As we neglect the ripples of the voltage across the EH load resistance $R_\text{L}$, the input voltage $v_\text{in}$ can be expressed as $v_\text{in} = i_\text{EH} R_{\Sigma}$, where $R_{\Sigma} = R_\text{L} + R_\text{s}$, see Fig.~\ref{Fig:MjCell}.
Thus, solving (\ref{Eqn:SolQuadEqn}) for $i_\text{EH}$, the output current $i_\text{EH}$ is given by (\ref{Eqn:EHCurrent_SingleJunction}) shown on top of the next page, where $W_0(\cdot)$ is the principal branch of the Lambert-W function.
Finally, the harvested power $P_\text{harv} = i^2_\text{EH} R_\text{L}$ can be obtained in closed form as a function of the received FSO power $p$ and current $p_\text{a}$ induced by the ambient light and is given by (\ref{Eqn:EHmodel}) shown on top of the next page.
\begin{figure*}[!t]
	\vspace*{0.1in}
	\begin{equation}
		i_{\text{EH}} = j + I_\text{s} - \frac{V_\text{T}}{R_{\Sigma}} W_0 \bigg( I_\text{s} \frac{R_{\Sigma}}{V_\text{T}}  \exp \bigg[\frac{R_{\Sigma}}{V_\text{T}} (j + I_\text{s}) \bigg] \bigg) 
		\label{Eqn:EHCurrent_SingleJunction}\vspace*{-7pt}
	\end{equation}
	\begin{equation}
		P_\text{harv}(p, p_\text{a}) = R_\text{L} \bigg[p r_0 + p_\text{a} + I_\text{s} - \frac{V_\text{T}}{R_{\Sigma}} W_0 \bigg( I_\text{s} \frac{R_{\Sigma}}{V_T}  \exp \Big[\frac{R_{\Sigma}}{V_\text{T}} (p r_0 + p_\text{a} + I_\text{s}) \Big] \bigg)\bigg]^2 \vspace*{-5pt}
		\label{Eqn:EHmodel}
	\end{equation}\hrulefill \vspace*{-15pt}
\end{figure*}

The derived EH model in (\ref{Eqn:EHmodel}) characterizes the harvested power at the photovoltaic RX for any values of the received FSO power $p$ and current $p_\text{a}$ induced by the ambient light.
Furthermore, for large received FSO powers, as $p \to \infty$, (\ref{Eqn:EHmodel}) asymptotically converges to the EH model derived in \cite{Wang2015}, where the EEC of the photovoltaic RX comprises only one diode, which models the diffusion current, and the recombination of charges in the depletion region of the p-n junction is neglected.

%% file: InfDecodingSection.tex
In the following, we discuss the information reception at the photovoltaic RX.
To this end, we first show that adopting the optimal filter for the output information signal may not be practical for non-linear photovoltaic RXs, and therefore, we design a suboptimal information RX.
Next, we derive the optimal transmit signal distribution maximizing the achievable information rate.
 
\paragraph*{Information RX} 
As in \cite{Tran2019, Wang2015}, we model the information RX by resistance $R_\text{d}$ and the current flow $i_\text{ID}(t)$ shown in Fig.~\ref{Fig:MjCell} is used to decode the received information signal.
Since only the \gls*{ac} component of the received FSO signal is utilized for information reception, the current flow $i_\text{ID}(t)$ can be expressed as \cite{Horowitz1989}:\vspace*{-4pt}
\begin{equation}
	i_\text{ID}(t) =  i_\text{ID}^\text{s} (t) + i_\text{ID}^\text{n} (t),
	\label{Eqn:IdCurrent}\vspace*{-4pt}
\end{equation}
where $i_\text{ID}^\text{s} (t)$ is the output current due to the received information signal $s(t)$ and depends on the parameters of the RX and the power spectral density of the received light $p_\text{r} (\lambda, t)$ \cite{Horowitz1989}.
Furthermore, $i_\text{ID}^\text{n} (t)$ is the equivalent noise at resistance $R_\text{d}$ that comprises the contributions of both the received noise $\tilde{w}(t)$ and the noise generated by the elements of the photovoltaic RX \cite{Mertens2014, Wang2015}.
We note that to demodulate the transmitted message, one can design a filter that is matched to the output information signal and yields the maximum achievable output \gls*{snr} after sampling \cite{Horowitz1989, Tse2005}.
However, the design and practical realization of such a filter may not be feasible due to the non-linearity of the photovoltaic RX, i.e., the non-linear dependency of $i_\text{ID}^\text{s} (t)$ on the spectral density of the received light $p_\text{r}(\lambda, t)$, and the high required sampling rate for an accurate discretization of the output signal $i_\text{ID}(t)$ \cite{Horowitz1989, Wang2015}.
Therefore, in the following, we resort to a suboptimal information RX.

To this end, we integrate the voltage $v_\text{d}(t) = i_\text{ID}(t) R_\text{d}$ across resistance $R_\text{d}$ over time slot $k, k\in\mathbb{N},$ and obtain:\vspace*{-3pt}
\begin{align}
	r[k] &=  \int_{(k-1)T}^{kT} R_\text{d} i_\text{ID}(t) \text{d} t = \int_{(k-1)T}^{kT} R_\text{d} C_\text{d} \frac{\text{d} v_\text{C}(t)  }{\text{d} t} \text{d} t
	\nonumber \\ &= \tilde{x}[k] - \tilde{x}[k-1] + \tilde{n}[k] - \tilde{n}[k-1],
	\label{Eqn:FilterOutputGeneral}\vspace*{1in}
\end{align}
\noindent where $v_\text{C}(t) = v^\text{s}_\text{C}(t) + v^\text{n}_\text{C}(t)$ is the voltage across the capacitance $C_\text{d}$ of the high-pass filter, as shown in Fig.~\ref{Fig:MjCell},~and $v^\text{s}_\text{C}(t)$ and $v^\text{n}_\text{C}(t)$ are the components of $v_\text{C}(t)$ corresponding to the received FSO signal $h s(t)$ and noise $w(t)$, respectively.
Furthermore, in (\ref{Eqn:FilterOutputGeneral}), $\tilde{x}[k] = R_\text{d} C_\text{d} v^\text{s}_\text{C}(kT)$ and $\tilde{n}[k] = R_\text{d} C_\text{d} v^\text{n}_\text{C}(kT)$ are the information signal and the noise in time slot $k$ at the output of the RX, respectively.

Since, in the steady state, the current flow $i_\text{ID}(kT) \approx 0, \forall k$, as shown Fig.~\ref{Fig:InpulseResponse}, voltage $v_\text{C}(kT)$ is equal to the voltage $i_\text{EH}[k] R_\text{L}$ across the EH load resistance $R_\text{L}$ in time slot $k$ \cite{Horowitz1989}, see Fig.~\ref{Fig:MjCell}.
Thus, we obtain the information symbol received in time slot $k, \forall k,$ as \vspace*{-3pt}
\begin{equation}
	\tilde{x}[k] = R_\text{d}C_\text{d} v^\text{s}_\text{C}(kT) = R_\text{d}C_\text{d} \sqrt{R_\text{L}} \sqrt{P_\text{harv}(h s[k], p_\text{a} )},\vspace*{-3pt}
\end{equation}
where $P_\text{harv}(\cdot, \cdot)$ is the power harvested at the EH load derived in Section~\ref{Section:EH} and given in (\ref{Eqn:EHmodel}).

We note that the output symbol $r[k], \forall k,$ in (\ref{Eqn:FilterOutputGeneral}) depends not only on $\tilde{x}[k]$ but also on $\tilde{x}[k-1]$.
To avoid the undesired memory of the photovoltaic RX, we assume that a sequence of symbols $\boldsymbol{s} = \{s[0], s[1], \cdots, s[K-1]\}$ of length $K \geq 1$ is transmitted.
Then, we obtain the normalized output symbol $y[k]$ in time slot $k, \forall k$, as follows\footnotemark:\vspace*{-5pt}
\begin{equation}
	y[k] = \frac{1}{R_\text{d}C_\text{d} \sqrt{R_\text{L}}} \sum_{p=0}^k r[p] = x[k] + n[k],
	\label{Eqn:CommChannel}\vspace*{-5pt}
\end{equation}
where $x[k] = \sqrt{P_\text{harv}(h s[k], p_\text{a} )}$ and $n[k] = \frac{1}{ R_\text{d}C_\text{d} \sqrt{R_\text{L}} } \tilde{n}[k]$ are the normalized output information signal and noise, respectively.
\footnotetext{We note that if the information RX is able to measure the voltage $v_\text{C}(t)$, the integration of the current flow $i_\text{ID}(t)$ in (\ref{Eqn:FilterOutputGeneral}) is not needed and the output symbol can be directly obtained as $y[k] = R_\text{L}^{-\frac{1}{2}} v_\text{C}(kT), \forall k$.}
Thus, unlike for optical RXs based on photodetectors in \cite{MaShuai2019, Ajam2020}, the output information signal $x[k]$ at the photovoltaic RX is not a linear function of the TX symbol $s[k], \forall k,$ but is determined by the non-linear function in (\ref{Eqn:EHmodel}).

We note that the output noise samples ${n}[k], \forall k \in\mathbb{N},$ include the impacts of the input RX noise and the thermal and shot noise of the photovoltaic RX \cite{Lapidoth2009}.
To model ${n}[k], \forall k$, we assume that the thermal noise at the RX dominates and the impact of transmit information symbol $s[k]$ on $n[k]$ is negligible \cite{Wang2015}.
Thus, we model ${n}[k], \forall k,$ as \gls*{iid} realizations of AWGN with variance $\sigma^2$ \cite{Lapidoth2009, Wang2015}.

\paragraph*{Achievable Information Rate} We note that in practical optical SWIPT systems, photovoltaic RXs can harvest power from different light sources, as, for example, sunlight, indoor illumination, and received FSO signals \cite{Mertens2014, Luque2010, MaShuai2019, Tran2019, Sepehrvand2021, Wang2015, Fakidis2020}.
However, information transmission to the optical non-linear photovoltaic RXs is challenging and constitutes a bottleneck in practical systems \cite{Sepehrvand2021, Wang2015, Fakidis2020}.
Therefore, in this paper\footnotemark, we design the optical SWIPT system assuming that the FSO TX aims at the maximization of the information rate between the TX and RX signals, while the RX also opportunistically harvests the power of the received FSO signal.
\footnotetext{We note that determining the tradeoff between the achievable information rate and the harvested power at the photovoltaic RX is an interesting direction for future work and is not tackled in this paper.}
In the following proposition, for a given maximum TX power $A^2$, we determine the optimal distribution of the transmit symbols $s$ for the maximization of the achievable rate of the considered FSO system employing a non-linear photovoltaic RX.
\begin{proposition}
	For a given cdf of the transmit information symbols, $F_s(s)$, an achievable rate in nats per channel use is given by\vspace*{-3pt}
	\begin{equation}
		R(F_s) = \frac{1}{2} \ln \big[ 1+ \frac{e^{2u(x; F_s)}}{ 2\pi e \sigma^2 } \big],
		\label{Eqn:RateBoundGeneral}\vspace*{-3pt}
	\end{equation}
	\noindent where $u(x; F_s)$ is the differential entropy of $x = \sqrt{P_\text{\upshape harv}(h s, p_\text{\upshape a} )}$ for a given $F_s(\cdot)$.
	Furthermore, the achievable rate in (\ref{Eqn:RateBoundGeneral}) is maximized by the following cdf of the transmit information symbols\vspace*{-3pt}
	\begin{equation}
		F^*_s(s) = \begin{cases}
			0, \quad & \text{\upshape if} \, s<0, \\
			 \frac{ P_\text{\upshape harv}(hs, p_\text{\upshape a}) - P_\text{\upshape harv}(0, p_\text{\upshape a}) }{ P_\text{\upshape harv}(hA^2, p_\text{\upshape a}) - P_\text{\upshape harv}(0, p_\text{\upshape a}) } , \quad &\text{\upshape if} \, s\in[0, A^2], \\
			1, \quad & \text{\upshape if} \, s>A^2,\vspace*{-3pt}
		\end{cases}
		\label{Eqn:CapacityAchievingDistributionS}
	\end{equation}
	\noindent and can be expressed as a function of $A^2$ as follows:\vspace*{-3pt}
	\begin{equation}
		\bar{R}(A^2) = \frac{1}{2} \ln \Big( 1 + \frac{ P_\text{\upshape harv}(hA^2, p_\text{\upshape a}) - P_\text{\upshape harv}(0) }{2 \pi e \sigma^2 } \Big). \vspace*{-3pt}
		\label{Eqn:RateBound}
	\end{equation}
	\label{Theorem:Capacity}\vspace*{-10pt}
\end{proposition}
\begin{proof}
	Please refer to Appendix~\ref{Appendix:ProofCapacity}.
\end{proof}\vspace*{-5pt}

Proposition~\ref{Theorem:Capacity} shows that in contrast to linear optical information RXs based on photodetectors \cite{MaShuai2019, Ajam2020}, uniformly distributed transmit symbols $s$ are not optimal for non-linear photovoltaic RXs.
Furthermore, we note that the cdf $F^*_s(\cdot)$ in (\ref{Eqn:CapacityAchievingDistributionS}), and thus, the achievable information rate in (\ref{Eqn:RateBound}) depend on the EH model in (\ref{Eqn:EHmodel}).
Moreover, since $P_\text{harv}(0, p_\text{a}) > 0$ for $p_\text{a} > 0$, the achievable information rate depends not only on the information symbols $s$, but also on the power spectral density $\tilde{p}_\text{a}(\lambda)$ of the ambient light.
Thus, an accurate modelling of optical RXs is important to characterize not only the harvested power, but also the performance of the information RX.

%% file: SimulationResults.tex
\begin{figure}[!t]
	\centering
	\includegraphics[draft = false, width=0.3\textwidth]{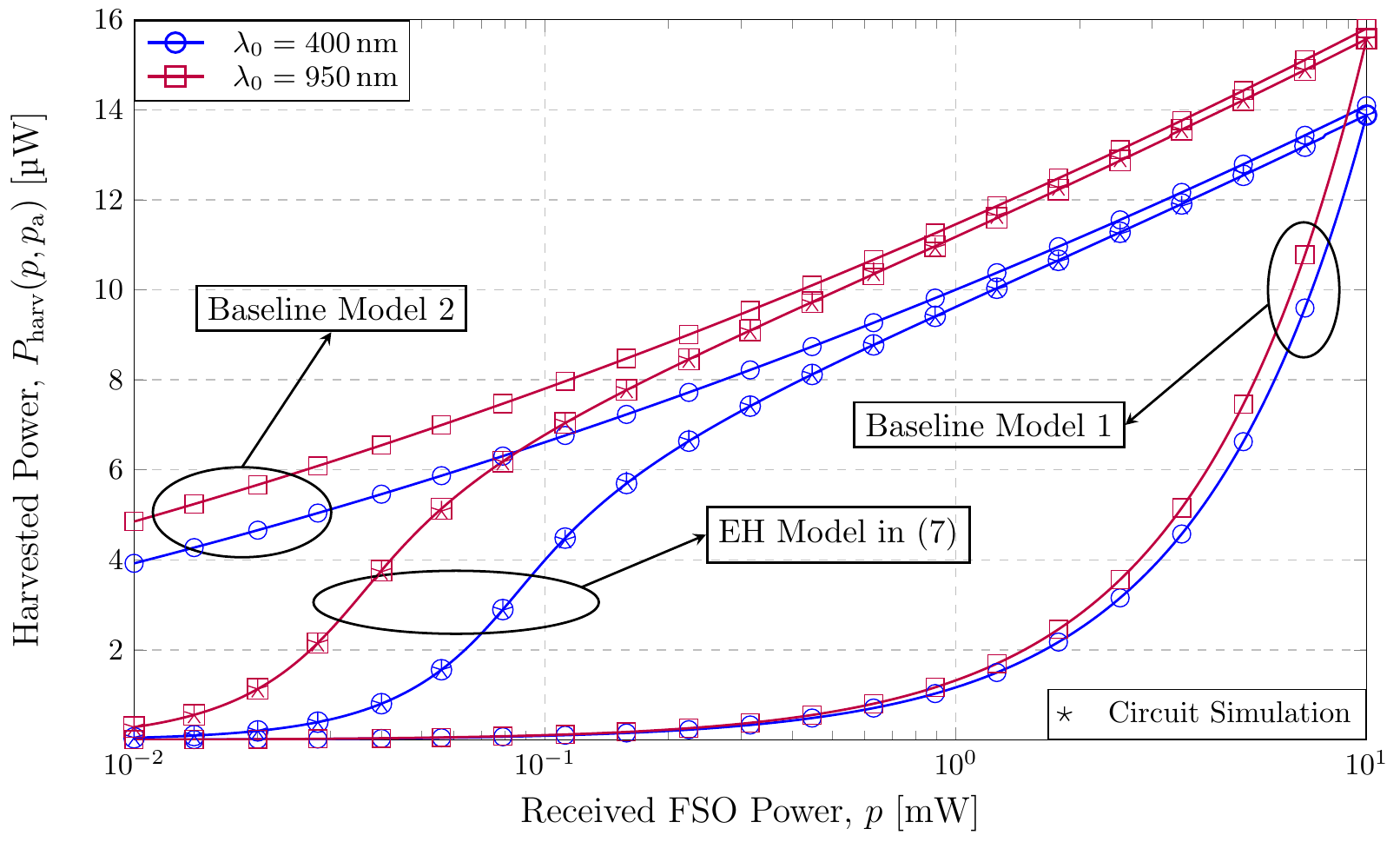}
	\vspace*{-5pt}
	\caption{Validation of the derived EH model via circuit simulations for different wavelength $\lambda_0$ of the transmitted FSO signal.}
	\label{Fig:ValidationEhModel}
	\vspace*{-10pt}
\end{figure}

In this section, we validate the derived EH model and evaluate the SWIPT system performance via simulations.
The adopted parameters of the photovoltaic RX circuit shown in Fig.~\ref{Fig:MjCell} are specified in Table~\ref{Table:SimulationSetup}.
The spectral response function of the photovoltaic RX in (\ref{Eqn:InducedCurrent}) is given by $r(\lambda) = \lambda \mu$, where $\mu = 0.7 \frac{q_0}{\eta c_l}$ is the conversion efficiency of the p-n junction at the RX \cite{Luque2010, Mertens2014}.
Here, constants $q_0$, $\eta$, and $c_l$ are the elementary charge, Planck's constant, and the speed of light, respectively.

In the following, we validate the non-linear EH model derived in Section~\ref{Section:EH} via circuit simulations.
To this end, in Fig.~\ref{Fig:ValidationEhModel}, we plot the harvested powers $P_\text{harv}(p, p_\text{a})$ obtained analytically by (\ref{Eqn:EHmodel}) for different values of $p$ and two wavelengths of the received FSO signal, $\lambda_0 \in \{\SI{400}{\nano\meter}, \SI{950}{\nano\meter}\}$.
Furthermore, for different values of $p$ and $\lambda_0$, we also simulate the RX EEC \cite{Mertens2014} shown in Fig.~\ref{Fig:MjCell} with the circuit simulation tool Keysight ADS \cite{ADS2017} and show the obtained harvested powers as star markers in Fig.~\ref{Fig:ValidationEhModel}.
Additionally, as Baseline Model 1 and Baseline Model 2, we adopt the EH models derived in \cite{MaShuai2019} and \cite{Wang2015} assuming MPP tracking at the RX and a single-diode EEC of the photovoltaic RX, respectively.
In particular, for Baseline Model 1, the model parameters are chosen to match the harvested power $P_\text{harv}(p, p_\text{a})$ in (\ref{Eqn:EHmodel}) for $p = \SI{10}{\milli\watt}$.
Finally, to investigate the EH performance of photovoltaic RXs not only in the high input power regime, but also when the received optical power is low, we set $p_\text{a} = 0$ in (\ref{Eqn:EHmodel}).

First, in Fig.~\ref{Fig:ValidationEhModel}, we observe that as expected, the harvested power $P_\text{harv} (p, p_\text{a})$ at the RX grows with the received FSO power $p$.
Next, we note that for all $p$ and $\lambda_0$, the results obtained with the derived EH model in (\ref{Eqn:EHmodel}) match the ADS circuit simulation results well.
Furthermore, since the spectral response of the photovoltaic cell $r_0$ grows linearly with $\lambda_0$, higher FSO wavelengths yield larger values of the harvested power in Fig.~\ref{Fig:ValidationEhModel}.
Also, we observe that the EH model in \cite{MaShuai2019} assumes a tunable EH load resistance $R_\text{L}$ to maximize the EH efficiency for a given received FSO power $p$, and thus, Baseline Model 1 is not able to capture the non-linearities of the photovoltaic RX since the parameters of the EH load are fixed.
Finally, for the EH model in \cite{Wang2015}, the EEC comprises only one diode to model the diffusion current of the p-n junction at the RX.
Thus, Baseline Model 2 can not accurately capture the RX non-linearities when the received FSO power is low and the harvested power is determined by the recombination of particles in the depletion region of the junction \cite{Luque2010}.

In Figs.~\ref{Fig:Plot_CDFs} and \ref{Fig:Plot_AchievableRates}, we investigate the performance of SWIPT systems with non-linear photovoltaic RXs.
To this end, in Fig.~\ref{Fig:Plot_CDFs}, we plot the cdfs $F^*_s$ in (\ref{Eqn:CapacityAchievingDistributionS}) maximizing the achievable information rate $R(F_s)$ in (\ref{Eqn:RateBoundGeneral}) for $p_\text{a} = 0$ and for different values of the maximum power of the transmit FSO signal $A^2$.
Furthermore, in Fig.~\ref{Fig:Plot_CDFs}, as a baseline scheme, we also plot the cdfs corresponding to uniformly distributed $s \in [0, A^2]$, which maximize the achievable rate if, as in \cite{MaShuai2019, Tran2019, Sepehrvand2021, Wang2015, Fakidis2020}, the non-linearity of the photovoltaic RX with respect to the received information signal is neglected.
Next, in Fig.~\ref{Fig:Plot_AchievableRates}, for normalized channels with $h = 1$, we plot the maximum achievable rates in (\ref{Eqn:RateBound}) for cdf $F_s^*$ and the achievable rates in (\ref{Eqn:RateBoundGeneral}) for uniformly distributed information symbols $s$ for different values of $A$ and $p_\text{a}$.
For the results in Fig.~\ref{Fig:Plot_AchievableRates}, we set the AWGN variance and FSO wavelength to $\sigma^2 = \SI{-60}{\dBm}$ and $\lambda_0 = \SI{950}{\nano\meter}$, respectively.

First, we observe in Fig.~\ref{Fig:Plot_AchievableRates} that since the harvested power in Fig.~\ref{Fig:ValidationEhModel} saturates for large input power values $p$, higher currents $p_\text{a}$ induced by the ambient light, and thus, higher intensities of the ambient light $\tilde{p}_\text{a}(\lambda)$, lead to lower achievable information rates $R(\cdot)$.
Next, we note that for all considered values of $A^2$ and $p_\text{a}$, the optimal cdf $F_s^*$ yields higher values of $R(\cdot)$ compared to a uniformly distributed TX signal.
Furthermore, while the maximum rate $\bar{R}(A^2)$ always increases with $A^2$, the achievable information rate for the uniformly distributed TX signal may even decrease when $A^2$ grows.
Thus, we conclude that for optical SWIPT systems, the accurate modelling of photovoltaic RXs is important not only for the precise characterization of the harvested power, but also for reliable communication.

%% file: ConclusionsSection.tex
In this work, we studied optical SWIPT systems, where a laser source emitted an intensity-modulated FSO signal focused on a photovoltaic RX.
First, by carefully analyzing the EEC of the non-linear photovoltaic RX, we derived an accurate closed-form EH model that characterizes the instantaneous harvested power at the RX.
Next, taking into account the non-linearity of the photovoltaic RX, we derived the optimal distribution of the transmit signal maximizing the achievable information rate between FSO TX and RX.
We validated the proposed EH model with circuit simulations.
Furthermore, we showed that in contrast to the derived EH model, two baseline EH models based on MPP tracking at the RX and a single-diode model of the p-n junction, respectively, are not able to capture the non-linearity of photovoltaic RXs in the low and high input power regimes, respectively.
We demonstrated that the achievable information rate decreases and increases when the intensity of the ambient light and the maximum intensity of the transmit FSO signal grow, respectively.
Finally, we observed that the derived optimal distribution yields significantly higher achievable information rates than uniformly distributed TX signals.

\begin{figure}[!t]
	\centering
	\includegraphics[draft = false, width=0.3\textwidth]{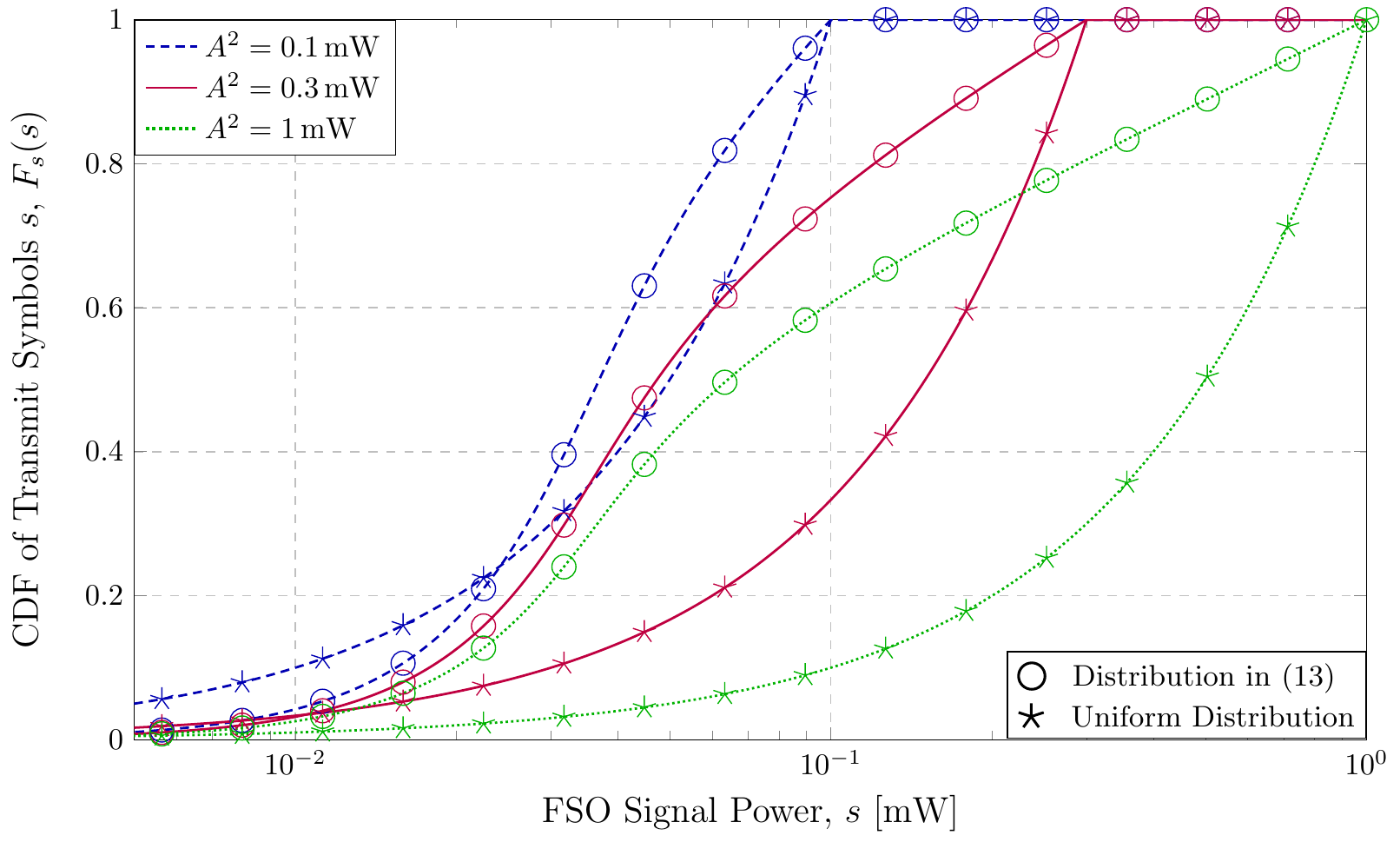}
	\vspace*{-7pt}
	\caption{Proposed cdf $F^*_s(s)$ in (\ref{Eqn:CapacityAchievingDistributionS}) and cdfs of uniformly distributed transmit signals $s$ for different values of $A^2$.}
	\label{Fig:Plot_CDFs}
	\vspace*{-13pt}
\end{figure}

%% file: ProofCapacity.tex
The proof follows similar steps as that of \cite[Theorem 5]{Lapidoth2009}.
Exploiting the entropy power inequality, for a given $F_s$, we obtain the mutual information $I(y, s; F_s)$ between the TX and RX symbols $s$ and $y$ as follows \cite{Lapidoth2009}:\vspace*{-3pt}
\begin{align}
	I(y , s ; F_s ) &=  u(y; F_s) - u(n) = u(x+n; F_s) - u(n) \nonumber\\
	& \geq \frac{1}{2} \ln \Big( 1+ \frac{e^{2u(x; F_s)}}{ 2\pi e \sigma^2 } \Big) \triangleq R(F_s),\vspace*{-3pt}
	\label{Eqn:RateBound2}
\end{align}
\noindent where $u(y; F_s)$ and $u(n) = \frac{1}{2} \ln(2 \pi e \sigma^2)$ are the differential entropies of $y$ for given $F_s$ and AWGN, respectively.

We note that for a given FSO TX power $A^2$, the symbols at the RX output are bounded by $x[k] \in [\sqrt{P_\text{harv}(0, p_\text{a})}, \sqrt{ P_\text{harv} (hA^2, p_\text{a}) } ], \forall k$.
Therefore, the differential entropy $u(x)$, and hence, the achievable rate in (\ref{Eqn:RateBound2}) are maximized if the pdf $f^*_x$ of $x$ is uniform in $\text{dom} \{f^*_x\} = [\sqrt{P_\text{harv}(0, p_\text{a})}, \sqrt{P_\text{harv} (hA^2, p_\text{a})}]$ and is given by \vspace*{-5pt}
\begin{equation}
	f^*_x(x) = \frac{1}{\sqrt{P_\text{harv} (hA^2, p_\text{a})} - \sqrt{P_\text{harv} (0, p_\text{a})}   }.\vspace*{-5pt}
	\label{Eqn:CapacityAchievingDistribution}
\end{equation}
Finally, the corresponding cdf of $s$ and the maximum achievable rate are given by (\ref{Eqn:CapacityAchievingDistributionS}) and (\ref{Eqn:RateBound}), respectively.
This concludes the proof.

\begin{figure}[!t]
	\centering
	\includegraphics[draft = false, width=0.3\textwidth]{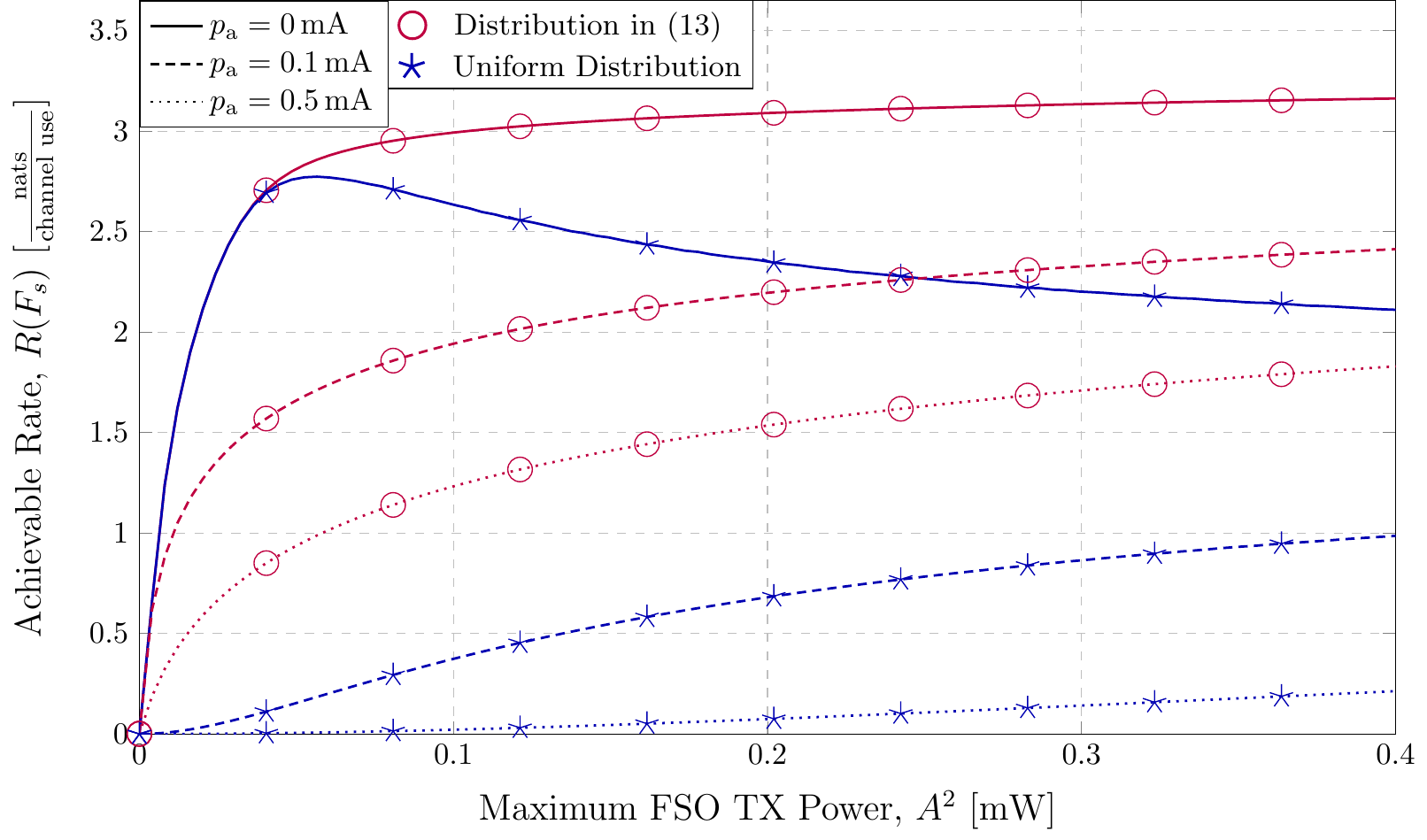}
	\vspace*{-7pt}
	\caption{Achievable information rates ${R}(F_\text{s})$ for different distributions $F_\text{s}$.}
	\label{Fig:Plot_AchievableRates}
	\vspace*{-13pt}
\end{figure}